\documentclass{article}
\usepackage{amsthm, amsmath}
\newtheorem{thm}{Theorem}
\newtheorem{lem}{Lemma}
\newtheorem{cor}{Corollary}
\usepackage{algpseudocode}
\algnewcommand{\Input}{\item[\textbf{Input:}]}
\algnewcommand{\Output}{\item[\textbf{Output:}]}
\newcommand{\LongState}[1]{\State
  \parbox[t]{\dimexpr\linewidth-\algorithmicindent-2em}{\hangindent=1.5em\relax
    #1\strut}}
\DeclareMathOperator{\Inf}{inf}

\newenvironment{ruledalg}[2]{\smallskip
  {\centering\textbf{#1}: #2\par}%
  \noindent\hrulefill
  \begin{algorithmic}}{\end{algorithmic}\unpenalty\unskip\penalty100000
  \hrulefill\par\smallskip}

\begin{document}

\title{Super-exponential query complexity reduction via
  noise-resistant quantum search} \author{Daniel Z. Zanger,
  Ph.D.\thanks{Email: danielzanger@gmail.com, Mailing address: 3133
    Connecticut Ave., NW \#521, Washington, DC 20008, USA.}}
\date{July 6, 2017}

\maketitle

\begin{abstract}
  In the SEARCH WITH ADVICE problem, a single entry of interest within
  a database of $N$ entries is to be found assuming that an ordering
  of the entries, from that with the highest probability of being the
  entry of interest (as determined by a so-called advice distribution)
  to that with the lowest, is provided. We present a quantum algorithm
  that, in the presence of significant levels of quantum noise, solves
  SEARCH WITH ADVICE for a power law advice distribution with
  average-case query complexity $O(1)$ as $N\rightarrow\infty$. Since
  as we also show the best classical algorithms for this problem
  exhibit average-case query complexity of order no better than
  $\log(N)$, our quantum algorithm provides a super-exponential
  reduction in query complexity.
\end{abstract}

 
\section{Introduction}

Well-known theoretical results have established that quantum
computation offers the promise of dramatic reductions in the
computational complexity of algorithms for data decryption (Shor
(1994)) and database search (Grover (1996, 1997)), among other critical
applications. Yet one of the key challenges for practical quantum
computation remains the sensitivity of quantum computations to errors
caused by ambient quantum noise. One of the most commonly suggested
ways of dealing with this problem is by means of fault-tolerant
quantum error correction (see for example Chapter~10 in Nielsen and
Chuang (2000)). However, this approach can be challenging to implement
as it is very expensive with respect to computational resources, and
large circuit sizes are in general required to apply it. Rather than
expending great effort to explicitly correct errors caused by quantum
noise, the approach we adopt in this paper is instead to accept their
presence in the context of quantum computations and to attempt to
reach a viable solution anyway. Under this as a guiding principle, we
analyze the computational complexity of quantum algorithms for
database search.

Assuming the absence of quantum noise, the renowned Grover quantum
search algorithm (Grover (1996,1997)), already cited above, finds a
single item of interest (called the ``marked element'') in a database
of $N$ items with just $O(\sqrt{N})$ database query calls as
$N\rightarrow\infty$ in the worst case, as compared with a worst-case
query complexity of $\Omega(N)$ for the best classical algorithms.  In
this article, however, we analyze the performance of quantum database
search algorithms in the presence of quantum noise. We study a
generalization, called SEARCH WITH ADVICE, of the standard,
unstructured search problem addressed by Grover. In SEARCH WITH ADVICE
(see \S 2 below), a so-called ``advice'' probability distribution
giving the probability of each item's actually being the marked
element, is presupposed, and a ``hint'', in the form of an ordering of
the $N$ items from highest advice probability to lowest, is actually
given as input to the search algorithm intended to be used to solve
the problem (so that SEARCH WITH ADVICE specializes to the standard
search problem under a uniform advice distribution).

In this article, the average-case query complexity of an algorithm is
defined as the expected number of queries to an oracle function
required by it (see also \S 2). Under a direct generalization of
depolarizing channel noise, with respect to which noise in the form of
a corresponding quantum channel occurs with probability $p=p(N)$, the
quantum algorithm we present here (as our Algorithm 2 in \S 3 below)
solves the SEARCH WITH ADVICE problem for advice distribution
following a power law with exponent $\frac{1}{N}-2$ with average-case
query complexity of just $O(1)$ as $N\rightarrow\infty$, given a noise
level as large as $p=p(N)=\frac{1}{\log(N)}$ (see our own Corollary~1
in \S4 below). This means that the expected query complexity is in
fact bounded by a fixed constant (not depending on $N$) for all
positive integers $N$ no matter how large, and this contrasts with a
minimal average-case query complexity no better than order $\log(N)$
(that is, $\Omega(\log(N))$ as $N\rightarrow\infty$) for the best
classical algorithm applied to the same problem (once again see
Corollary 1 below). Hence the quantum search algorithm we introduce in
this paper achieves a dramatic super-exponential reduction in
computational complexity over the best possible classical algorithm
for this search problem, even in the presence of quantum noise levels
declining quite slowly as $N\rightarrow\infty$. Quantum supremacy
(Preskill (2013)) refers to an empirical demonstration that a quantum
processor can perform some computational task faster than any
classical computer. Offering such a dramatic quantum computational
speed-up over corresponding classical algorithms as well as being as
robust as it is with respect to quantum noise, our quantum algorithm
for solving SEARCH WITH ADVICE appears, for example, to possess
significant potential as a means to establishing practical quantum
supremacy.

Our database search algorithm here is a combination of a quantum
database search algorithm in Montanaro (2010), which solves the SEARCH
WITH ADVICE problem without addressing the possibility of quantum
noise, and that in Vrana et al. (2014), which solves the standard,
unstructured database search problem in the presence of quantum
noise. In fact, by comparison, the algorithm in Vrana et al. (2014)
(see in particular their Algorithm 3 in their \S 2.4) exhibits a
worst-case query complexity no better than order $\frac{N}{\log(N)}$
(that is, a query complexity of
$\Omega\left(\frac{N}{\log(N)}\right)$) for the standard database
search problem (with no advice distribution) in the context of the
same noise model as here with noise probability once again
$p=p(N)=\frac{1}{\log(N)}$. This is not nearly as dramatic a
computational speed-up relative to the classical case as we are able
to achieve here in the average-case complexity setting for SEARCH WITH
ADVICE. Rather similarly to the results in Vrana et al.~(2014) and
once again for the search problem without advice distribution, the
authors of Cohn et al. (2016) have deduced that the maximum
depolarizing channel noise probability possible to achieve any
computational advantage over the fastest classical algorithms is
$p=p(N)=\Omega\left(\frac{\log(\sqrt{N})}{\sqrt{N}}\right)$ as
$N\rightarrow\infty$ (see \S~III.B.1 in Cohn et al. (2016)), which is
a much smaller level of quantum noise than we allow to gain a much
greater quantum speed-up in the average case for power law
distribtions under SEARCH WITH ADVICE. Moreover, in Corollary~2 in
\S4, we also establish for the same problem that the amount of quantum
noise can even be increased to order
$p=p(N)=\left(\frac{1}{\log(N)}\right)^{q}$ for any $q$, $0<q<1,$
while still solving the problem with query complexity of order no
larger than $(\log(N))^{1-q}$. Hence an at least polynomial speedup
with respect to the best classical algorithms is still achieved in
this case.

The paper is organized as follows. In the next section we formally state the SEARCH WITH ADVICE problem and introduce some necessary background concepts and notation. In \S 3 we present our noise-resistant quantum search algorithm for solving this problem. Our results on the query complexity of this quantum algorithm and corresponding classical algorithms for solving the same problem are stated and proved in \S 4.

\section{The Search with Advice Problem}
We now state the SEARCH WITH ADVICE problem.

\begin{ruledalg}{Problem}{SEARCH WITH ADVICE}
  \Input An oracle function $f:\{1,...,N\}\rightarrow \{0,1\}$ on
  $N=2^{q}$ elements, for any positive integer $q$, that takes the
  value $1$ on precisely one element $n_{me}\in\{1,...,N\}$, and an
  ``advice'' probability distribution $\mu=(\mu_{n}),n=1,...,N$, where
  $\mu_{n}$ is an assessed probability that in fact $n=n_{me}.$

  \Output The unique element $n_{me}$, called the ``marked element'',
  for which $f(n_{me})=1$.  
\end{ruledalg}

In this paper our interest is in analyzing and comparing algorithms ---
both classical and quantum --- that solve the SEARCH WITH ADVICE
problem above with minimal query complexity. In the context of SEARCH
WITH ADVICE, the query complexity (Montanaro (2010)) of a classical or
quantum algorithm is the number of queries to the oracle function $f$
in the statement of the problem that is required to identify the
marked element. Consider any quantum or classical algorithm $\cal{A}$,
which, given access to the oracle $f$, is designed to solve SEARCH
WITH ADVICE by identifying the marked element. Indeed we call
$\cal{A}$ a \emph{valid} algorithm if it outputs the marked element
with certainty. Let ${\cal{D}}$ denote the class of valid
deterministic classical algorithms.

We intend to investigate the average-case query complexity (that is,
here, the expected query complexity) of efficient quantum and
classical algorithms for solving SEARCH WITH ADVICE. So, let
$T_{\cal{A}}(\mu)$ denote the expectation of the number of queries to
$f$ used by $\cal{A}$, where this expectation is taken over the
distribution $\mu$ and (potentially) the internal randomness of
$\cal{A}$. That is, this average-case query complexity
$T_{\cal{A}}(\mu)$ is defined by
\begin{equation}
T_{\cal{A}}(\mu)=\sum^{N}_{n=1}\mu_{n}T_{\cal{A}}(n),
\end{equation}
where, in turn, $T_{\cal{A}}(n)$ is defined as the expectation of the
number of queries to $f$ used by $\cal{A}$ to identify $n$ as the
marked element, if in fact $n$ is the marked element. We also define a
key corresponding quantity of interest, the classical (that is,
classical algorithm) average-case query complexity of $\mu$, as
\begin{equation}
D(\mu)=\Inf_{\cal{A}\in {\cal{D}}}T_{\cal{A}}(\mu).
\end{equation}

We assume in this paper that the advice distribution $\mu$ is
non-increasing, so that $\mu_{n_{1}}\geq \mu_{n_{2}}$ whenever
$n_{1}\leq n_{2}$. With this assumption, the optimal classical
algorithm to find $n_{me}$ is clearly to query $f(1)$ through $f(N)$
in turn, so the classical average-case query complexity is easily seen
in this case to be
\begin{equation}
\sum^{N}_{n=1}\mu_{n}n=D(\mu)=\min_{\cal{A}\in D}T_{\cal{A}}(\mu).
\end{equation}

\section{Noise-resistant geometric quantum search algorithm}

The idea behind the quantum search algorithm we present in this
section is to combine an algorithm in Vrana et al. (2014), which is
designed to be robust against quantum noise, with one in Montanaro
(2010), which achieves super-exponential expected computational
advantage over classical algorithms in the absence of quantum
noise. Algorithm 2 in this paper achieves super-exponential expected
computational advantage over the optimal classical algorithm in the
presence of significant levels of quantum noise. The algorithms in
Montanaro (2010) and Vrana et al. (2014) are both ultimately based on
the original quantum search algorithm of Grover (Grover (1996, 1997)),
so in that sense, our own quantum search algorithm here for SEARCH
WITH ADVICE, which we present in this section, is as well.

To describe our algorithm, for any positive integer
$N_{1},1\leq N_{1}\leq N$, let $\mathbf{C}^{d_{1}}$ be a
quantum state space of dimension
\begin{equation}
d_{1}=d_{1}(N_{1})=2^{(\mbox{\small min}([\log_{2}(N_{1})]+1,q))},
\end{equation}
where we recall that by definition $N=2^{q}$ for some positive integer
$q$. Note that this definition implies that
$N_{1}\leq d_{1}=d_{1}(N_{1})\leq 2N_{1}$.  Define the action of the
quantum oracle operator on a corresponding computational basis, which
we enumerate as $|n\rangle,n=1,...,d_{1},$ by the unitary map
\begin{equation}
 {\cal{O}}_{f}:\mathbf{C}^{d_{1}}\rightarrow
 \mathbf{C}^{d_{1}},\,|n\rangle \mapsto (-1)^{f(n)}|n\rangle. 
\end{equation}
Note of course that ${\cal{O}}_{f}(|n\rangle)=|n\rangle$ only at $n=n_{me}$ (and
of course we still assume there is exactly one marked element in the
entire set $\{1,...,N\}$). Let
\begin{equation}
|\psi\rangle=\frac{1}{d_{1}^{1/2}}\sum^{d_{1}}_{n=1}|n\rangle
\end{equation}
be the equal superposition state, which can be generated using the
Hadamard transform (see Nielsen and Chuang (2000), Chapter 6), and
also define a corresponding unitary operator via
\begin{equation}
{\cal{U}}_{|\psi\rangle}=I-2|\psi\rangle\langle\psi|,
\end{equation}
where $I$ is of course the identity operator. In addition, we require
a suitable model of quantum noise, and we consider a generalization
(as in Vrana et al. (2014)) of the depolarizing channel. So let $T$ be
any arbitrary, given quantum channel (quantum operation) acting on
density operators $\rho$ on the state space $\mathbf{C}^{d_{1}}$, and
define a quantum noise model, parametrized by a probability value
$p\in [0,1]$, via
\begin{equation} {\cal{N}}_{p}(\rho)=(1-p)\rho+pT(\rho).
\end{equation}
When $T=\frac{I}{d_1}$, (8) is of course the standard
$d_{1}$-dimensional depolarizing channel (see Nielsen and Chuang
(2000)).

We now state Subroutine~1, which will be called by Algorithm 1 below
and hence by our main quantum search algorithm (Algorithm 2) for
SEARCH WITH ADVICE. Subroutine 1 is in essence the Grover iteration
step appearing in the standard version of the Grover algorithm but
with quantum noise present (see also Vrana et al. (2014) or for, for
the version of the Grover iteration step without noise, Chapter 6 in
Nielsen and Chuang (2000)).

\begin{ruledalg}{Subroutine~1}{Grover iteration with quantum noise}

  \Input   The oracle function $f$ from the statement of the SEARCH WITH ADVICE
  problem above;\, a positive integer $N_{1}$ with $1\leq N_{1}\leq
  N$; a density operator
  $\rho:\mathbf{C}^{d_{1}}\rightarrow\mathbf{C}^{d_{1}}$ where
  $d_{1}=d_{1}(N_{1})=2^{(\min([\log_{2}(N_{1})]+1,q))}$;
  a desired number of iterations $M$. Moreover, assume that the equal
  superposition state $|\psi\rangle$ as in (6) has been prepared and
  that quantum noise is modeled as in (8) above for some $p\in
  [0,1]$.

  \Output The density operator resulting from application of $M$
  Grover iterations with noise to $\rho$.

  \State $\mathit{count} :=1$;
  \State $\mathit{Groverstep} := \rho$;

  \While {$\mathit{count} \le M$}
  \State $\mathit{Groverstep} :={\cal{U}}_{|\psi\rangle}({\cal{O}}_{f}
  ({\cal{N}}_{p}(\mathit{Groverstep})){\cal{O}}^{\dagger}_{f}){\cal{U}}^{\dagger}_{|\psi\rangle}$; 
  \State $\mathit{count} := \mathit{count}+1$;
  \EndWhile
  \State \Return $\mathit{Groverstep}$;
\end{ruledalg}

Exploiting the basic Grover iteration above, we now state a version of
the Grover quantum search algorithm (Algorithm 1 below) which will be
invoked within our noise-resistant geometric quantum search algorithm
(Algorithm 2) below. Algorithm 1 here has itself appeared as Algorithm
1 in \S 2.4 of Vrana et al. (2014). In order to state it, define, for
any real numbers $\epsilon, c > 0$ and all nonnegative integers $i$,
the value 
\begin{equation}
\alpha_{i}(\epsilon)=\frac{1}{\sqrt{1+\frac{i}{c\log(1/\epsilon)}}}.
\end{equation}

\begin{ruledalg}{Algorithm 1}{Noise-resistant Grover search}
\Input The function $f$ from the statement of the SEARCH WITH ADVICE
problem above; two integers $n_{1},n_{2}\in\{1,...,N\},$ indicating
where a search of some subset of consecutive numbers from among the
set $1,...,N$ is to begin and end, respectively, inclusive of the two
numbers $n_{1},n_{2}$; two adjustable parameters $\epsilon,
c>0$. Assume as well that quantum noise (corresponding to some $p\in
[0,1]$) affects Algorithm 1 through its presence in Subroutine 1, and
define $N_{1}=n_{2}-n_{1}+1$. 

\Output:  The marked element $n_{me}$ if found; otherwise $0$.

\For{$i=0$, $1$, $2$, \ldots}
\LongState{1.  Prepare the equal superposition state $|\psi\rangle=\frac{1}{d_{1}^{1/2}}\sum^{d_{1}}_{n=1}|n\rangle$ on a quantum register $\mathbf{C}^{d_{1}}$,
where $d_{1}=d_{1}(N_{1})$ is as in (4);}
\LongState{2. Let $\rho=\frac{1}{d_{1}}\sum^{d_{1}}_{n=1}|n\rangle\langle
n|$, and apply Grover iteration with quantum noise (Subroutine 1)
above with 
inputs $f$, $N_{1}$, $\rho,|\psi\rangle$, and
$M=\left[\alpha_{i}(\epsilon,c)\frac{\pi}{4}\sqrt{d_{1}}\right]$;}
\LongState{3. Measure $\min([\log_{2}(N_{1})]+1,q)$ qubits in the standard
basis, and check the result using one oracle invocation;}
\EndFor
\If{$n_{me}$ found}
\State\Return $n_{me}$
\Else
\State\Return $0$
\EndIf
\end{ruledalg}

The geometric quantum search algorithm of Montanaro (2010) (Algorithm
1 in \S 2.1 there), on which Algorithm 2 below is in part based, does
not incorporate the possibility of quantum noise as is done here in
Subroutine 1 and Algorithm 1 as above. Our own noise-resistant
geometric quantum search algorithm (Algorithm 2) for SEARCH WITH
ADVICE, which does incorporate quantum noise, in essence merges the
algorithms of both Montanaro (2010) and Vrana et
al. (2014). Informally Algorithm 2 consists in partitioning
the input into successive blocks which increase in size geometrically
(hence the algorithm's name) and performing Algorithm 1 on each of
these blocks. We are now in a position to state it.

\begin{ruledalg}{Algorithm 2}{Noise-resistant geometric quantum search}
  \Input The oracle function $f:\{1,...,N\}\rightarrow\{0,1\}$ from
  the SEARCH WITH ADVICE problem as above; Advice distribution $\mu
  =(\mu_{n})$ (though all that is actually needed to implement the
  algorithm is just the ordering of the numbers $1,...,N$ as
  $n_{1},...,n_{N},$ where the $n_{i}\in\{1,...,N\}$ for all
  $i=1,...,N$, the $n_{i}$ are all distinct, and $\mu_{n_{i}}\geq
  \mu_{n_{i+1}},i=1,...,N-1$); chosen real values $\epsilon, c >0$.
  Assume as well that quantum noise is present in Algorithm 2 through
  its appearance in Subroutine 1 (for some $p\in [0,1]$).

  \Output The marked element $n_{me}$.

  \State $\mathit{start} :=1$;
  \State $\mathit{end} :=1$;
  \State $\mathit{step} :=0$;

  \While {$\mathit{start} \le N$}
  \If{$\mathit{end} - \mathit{start} \ge 100$}
  \LongState{Perform Noise-resistant Grover search (Algorithm 1) with
  $\epsilon,c >0$ to identify $n_{me}$ or its absence in the subset
  $\{\mathit{start}, \ldots, \mathit{end}\}$ (where $n_{1}=start$ and 
  $n_{2}=\mathit{end}$ in the notation of Algorithm 1);}
  \Else
  \LongState{Perform classical (non-quantum) search (as in \S 2) to
  identify $n_{me}$ or its absence within the set $\{\mathit{start},
  \ldots, \mathit{end}\}$;}
  \EndIf
  \If{$n_{me}$ was found}
  \State\Return{$n_{me}$}
  \Else
  \State $\mathit{step}:= \mathit{step}+1$;
  \State $\mathit{start} := \mathit{end}+1$;
  \State $\mathit{end} := \min(\mathit{start}+[e^{\mathit{step}}]-1,N$); 
  \EndIf
  \EndWhile
\end{ruledalg}

\section{Query Complexity Results}

Now for any valid quantum or classical algorithm ${\cal{A}}$ as at the
beginning of \S 2, recall, for any advice measure $\mu$, the
average-case query complexity values $T_{\cal{A}}(\mu)$ and
$T_{\cal{A}}(n),n=1,...,N,$ from (1). We denote by ${\cal{A}}_{QS}$
our Algorithm 2 in \S 3, which is a valid algorithm, and we can
consider the associated quantities  $T_{{\cal{A}}_{QS}}(\mu)$ and
$T_{{\cal{A}}_{QS}}(n),n=1,...,N.$ Also, we can write
${\cal{A}}_{QS}={\cal{A}}_{QS}(p)$ to make the ambient quantum noise
level $p\in [0,1]$ (as in Subroutine 1) explicit. 

The following Theorem 1 is a direct consequence of Theorem 3 in \S 2.4
of Vrana et al. (2014), along with the basic observation (see \S 1.2
in Montanaro (2010) or Motwani and Raghavan (1995), Exercise 1.3 in \S
1.2 there) that, given a (classical or quantum) search algorithm
$\cal{A}$ that uses $k$ query calls and outputs the marked element
$n_{me}$ with probability $s$, there is a classical algorithm
${\cal{A}}_{1}$ that takes ${\cal{A}}$ and $s$ as inputs and outputs
the marked element with certainty, doing so using an expected number
of queries of at most $\frac{k+1}{s}$. In the statement of Theorem 1
below we analyze the quantum algorithm ${\cal{A}}$ given as Algorithm
1 in \S 3 above, and, as the observation just mentioned is a standard
one from the theory of randomized algorithms, we will for any value
probability value $s$ simply identify the algorithm ${\cal{A}}$ given
in \S 3 as Algorithm 1 with the algorithm
${\cal{A}}_{1}={\cal{A}}_{1}({\cal{A}},s)$ and can use the same
notation to refer to both of them.

\begin{thm}[Vrana et al. (2014)]
  With respect to the above Algorithm 1 (for which in this theorem
  statement we use the notation ${\cal{A}}$), let $N_{1}$ be any
  integer with $100\leq N_{1}\leq N,$ and also let
  $d_{1}=d_{1}(N_{1})$ be as in (4). Suppose as well that
  $\epsilon\in (0,\frac{1}{2}]$ is given. Furthermore, assume that
  $p=p(N),$ which may depend on $N$, is any value $p\in [0,1]$ which
  describes the ambient quantum noise level via
  ${\cal{N}}_{p}(\rho)=(1-p)\rho+pT(\rho)$, where in turn $\rho$ is
  any density operator acting on the state space
  $\mathbf{C}^{d_{1}}$ and $T$ is a given quantum channel acting on
  the corresponding quantum register when executing Algorithm 1. Then,
  Algorithm 1 with $c=10$ and $\epsilon$ as given (or, more
  technically, the algorithm
  ${\cal{A}}_{1}={\cal{A}}_{1}({\cal{A}},1-\epsilon)$ as discussed in
  the previous paragraph) finds the marked element with certainty (if
  it is present within the subset of $\{1,...,N\}$ of size $N_{1}$
  being searched by Algorithm 1) after an expected number of not more
  than
\begin{equation}
\left(\frac{100}{1-\epsilon}\right)\left(1.02+d_{1}p+\sqrt{d_{1}}\right)\log\left(\frac{1}{\epsilon}\right) 
\end{equation}
oracle queries, where $d_{1}\leq 2N_{1}$. 
\end{thm}

For an arbitrary advice distribution, Theorem 1 enables us to now
state and prove Lemma 1. Lemma 1 --- our first new result --- is a bound
on the expected query complexity of Algorithm 2 in \S 3.

\begin{lem}
  Let $N$ be any positive integer, and assume that
  $\epsilon\in (0, \frac{1}{2}]$ and $p\in [0,1]$, where $p=p(N)$ may
  depend on $N$. Then the expected number of queries used by Algorithm
  2 in \S 3 for any given advice distribution $\mu=(\mu_{n})$ is
  upper-bounded by
\begin{equation}
T_{{\cal{A}}_{QS}(p)}(\mu)\leq e^{2}\sum^{N}_{n=1}\mu_{n}G(n,p,\epsilon),
\end{equation}
where
\[
  G(r,p,\epsilon)=(\frac{200}{1-\epsilon})(1.04+rp+\sqrt{r})\log(\frac{1}{\epsilon})
\]
for $r\in [0,\infty),p\in [0,1],\epsilon\in (0,\frac{1}{2}]$. 
\end{lem}

\begin{proof}
   In the $m$th repetition of the loop, the (at most) $[e^{m}]$
   elements contained in the range 
\begin{equation}
R_{m}=\{1+\sum^{m-1}_{i=0}[e^{i}],...,\min([e^{m}]+\sum^{m-1}_{i=0}[e^{i}],N)\}
\end{equation}
will be searched. By Theorem 1 with $c=10$, the noisy Grover search
step in this iteration uses at most an expected
\begin{displaymath}
\left(\frac{100}{1-\epsilon}\right)\left(1.02+2[e^{m}]p+\sqrt{2[e^{m}]}\right)\log\left(\frac{1}{\epsilon}\right)+1
\end{displaymath}
 number of queries. So, an expected total of at most
\begin{equation}
\sum^{m}_{i=0}\left( \left(\frac{100}{1-\epsilon}\right)\left(1.02+2[e^{m}]p+\sqrt{2[e^{m}]}\right)\log\left(\frac{1}{\epsilon}\right)+1\right)\leq \sum^{m}_{i=0}G(e^{i},p,\epsilon)
\end{equation}
queries will be used by Algorithm 2 to search for the marked element up to and including the $m$th repetition of the loop. But it is clear from (12) that, for any $n\in R_{m},\,m\leq \log_{e}(n)+1$. The average-case query complexity is therefore upper-bounded by
\begin{equation}
\sum_{n=1}^{N}\mu_{n}\left(\sum^{[\log_{e}(n)+1]}_{i=0}G(e^{i},p,\epsilon)\right),
\end{equation}
and, estimating the inner sum above by an integral, we obtain an upper bound of 
\begin{equation}
\sum^{N}_{n=1}\mu_{n}\left(\int^{\log_{e}(n)+2}_{0}G(e^{s},p,\epsilon)ds\right)\leq\sum^{N}_{n=1}\mu_{n}\left(e^{2}G(n,p,\epsilon)\right),
\end{equation}
completing the proof.
\end{proof}

We use Lemma 1 to prove Theorem 2, which we now state. Theorem 2
asserts that, for advice distributions defined by certain power law
distributions, Algorithm 2 as presented in the previous section
achieves greater-than-exponential speed-ups in average-case query
complexity for SEARCH WITH ADVICE, relative to the best possible
classical algorithms for this problem. Theorem 2 and its Corollaries 1
and 2, which also follow below, are the main results of the paper.

\begin{thm}
  For any positive integer $N\geq 100$ and any $\delta\in
  (0, \frac{1}{4}]$, define an advice distribution $\mu=\mu_{\delta}$
  on $\{1,...,N\}$ by taking
  $\mu_{\delta,n}=\alpha_{(\delta-2)}n^{(\delta-2)}$ for each
  $n\in\{1,...,N\}$, where
  $\frac{1}{\alpha_{(\delta-2)}}=\sum^{N}_{n=1} n^{(\delta-2)}.$ Then,
  $D(\mu_{\delta})\geq \frac{3N^{\delta}-3}{8\delta}$, but, for any
  $p\in [0,1],\,T_{{\cal{A}}_{QS}(p)}(\mu_{\delta})\leq
  400e^{2}\left(c_{1}+\frac{c_{2}p(N^{\delta}-1)}{\delta}\right)$, for
  some constants $c_{1},c_{2}>0$, where the level $p=p(N)$ of quantum
  noise as in (8) may depend on $N$. Furthermore, we note that we can
  take $c_{1}=9$ and $c_{2}=1.04.$ 
\end{thm}

\begin{proof}
  For any desired $r, -2<r<-1,$ define a probability measure on
  $\{1,...,N\}$ via $\mu_{n}=\alpha_{r}n^{r},n=1,...,N.$ Since $\mu$
  is prescribed to be a probability measure, we have that
  $\frac{1}{\alpha_{r}}=\sum^{N}_{n=1} n^{r}.$ This sum can be
  estimated by an integral, giving 
\begin{equation}
\int^{N}_{1}x^{r}dx\leq\frac{1}{\alpha_{r}}\leq 1+\int^{N}_{1}x^{r}dx.
\end{equation}
This implies that 
\begin{equation}
\frac{N^{r+1}-1}{r+1}\leq\frac{1}{\alpha_{r}}\leq\frac{N^{r+1}-1}{r+1}+1.
\end{equation}
We have, for $r=\delta-2$,
\begin{multline}
  D(\mu)=\alpha_{r}\sum^{N}_{n=1}n^{r+1}\geq\alpha_{r}\int^{N}_{1}x^{r+1}dx\\
  \geq\frac{(N^{r+2}-1)(r+1)}{(N^{r+1}+r)(r+2)}=\frac{(N^{\delta}-1)(\delta-1)}{(N^{\delta-1}+\delta-2)(\delta)}.
\end{multline}
Since $|\delta-1|\geq\frac{3}{4}$ and $|N^{\delta-1}+\delta-2|\leq 2$
for $\delta\in (0,\frac{1}{4}]$, we obtain the lower bound on $D(\mu)$
in the statement of the theorem. Again for $-2<r<-1$ and applying
Lemma 1 with any $\epsilon\in (0,\frac{1}{2}]$,
\begin{multline}
T_{{\cal{A}}_{QS}(p)}(\mu) \leq e^{2}\sum^{N}_{n=1}\mu_{n}G(n,p,\epsilon)\\
=e^{2}\alpha_{r}\sum^{N}_{n=1}n^{r}G(n,p,\epsilon)\\
\leq e^{2} \alpha_{r}\left(G(1,p,\epsilon)+\int^{N}_{1}x^{r}G(x,p,\epsilon)dx\right)\\
\leq
200(1-\epsilon)^{-1}e^{2}\log(\epsilon^{-1})\left(\frac{r+1}{N^{r+1}-1}\right)\\
\times
\left(2.04+p+\frac{1.04N^{r+1}-1.04}{r+1}+\frac{p N^{r+2}-p}{r+2}+\frac{N^{r+\frac{3}{2}}-1}{r+\frac{3}{2}}\right)\\
\leq
200(1-\epsilon)^{-1}e^{2}\log(\epsilon^{-1})\left(1.04+1.04\left(7.04+\frac{p(N^{\delta}-1)}{\delta}\right)\right), 
\end{multline}
having chosen $r=\delta-2$. Taking $\epsilon=\frac{1}{2}$ now leads to
the stated result.
\end{proof}

Corollary 1 which we now present shows that, for quantum noise levels
as large as $p=p(N)=\frac{1}{\log(N)}$ and advice distribution
following a power law, the best possible classical algorithm for
SEARCH WITH ADVICE in this case has expected query complexity growing
no more slowly than a rate of order $\log(N)$ as $N\rightarrow\infty$,
whereas the query complexity of our geometric quantum search
algorithm, Algorithm 2, is in fact bounded by a fixed constant (not
depending on $N$) for all positive integers $N$. 

\begin{cor}
  For any value $\delta$ with $0<\delta\leq\frac{1}{N}$, define an advice
  distribution $\mu_{\delta}$ on $\{1,...,N\}$ by taking
  $\mu_{\delta,n}=\alpha_{(\delta-2)}n^{(\delta-2)}$ for each
  $n\in\{1,...,N\}$ where
  $\frac{1}{\alpha_{(\delta-2)}}=\sum^{N}_{n=1}
  n^{(\delta-2)}$. Suppose that the ambient quantum noise level as in
  (8) is described by $p=p(N)=\frac{1}{\log(N)}$. Then,  
\begin{equation}
  D(\mu_{\delta})\geq \Omega(\log(N)),\quad\text{but}\quad
  T_{{\cal{A}}_{QS}(p(N))}(\mu_{\delta})\leq O(1).
\end{equation}
In fact, for all $N\geq 100$,
\begin{equation}
  D(\mu_{\delta})\geq c_{3}\log(N),\quad\text{and}
  \quad T_{{\cal{A}}_{QS}(p(N))}(\mu_{\delta})\leq c_{4},
\end{equation}
for constants $c_{3},c_{4}>0,$ where we can take $c_{3}=\frac{3}{8}$ and
$c_{4}=4440e^{2}<32810$. 
\end{cor}

\begin{proof}
  The first derivative of the function
  $g(\delta)=N^{\delta},\delta\in\mathbf{R},$ is the
  function $g'(\delta)=\log(N)N^{\delta}.$ Hence, by the Mean Value
  Theorem, $\frac{N^{\delta}-1}{\delta}=\log(N)N^{\delta_{1}}$, where
  $0\leq\delta_{1}\leq\delta\leq \frac{1}{N}$. But, for such
  $\delta_{1}$, $1\leq N^{\delta_{1}}\leq N^{1/N}\leq 2$. So,
  taking $p=p(N)=\frac{1}{\log(N)}$ in the statement of Theorem 2
  establishes the result. 
\end{proof}

In Corollary 2, we extend Corollary 1 by allowing increased levels of
quantum noise, at the expense of obtaining less dramatic reductions in
query complexity by means of quantum search.

\begin{cor}
  For any value $\delta$ with $0<\delta\leq\frac{1}{N}$, define an advice
  distribution $\mu_{\delta}$ on $\{1,...,N\}$ by taking
  $\mu_{\delta,n}=\alpha_{(\delta-2)}n^{(\delta-2)}$ for each
  $n\in\{1,...,N\}$ where
  $\frac{1}{\alpha_{(\delta-2)}}=\sum^{N}_{n=1}
  n^{(\delta-2)}$. Suppose that the ambient quantum noise level as in
  (8) is described by $p=p(N)=\frac{1}{\left(\log(N)\right)^{q}},$ for some
  $q,0<q\leq 1$. Then,  
\begin{equation}
  D(\mu_{\delta})\geq \Omega(\log(N)),\quad\text{but}\quad
  T_{{\cal{A}}_{QS}(p(N))}(\mu_{\delta})\leq O((\log(N))^{1-q}).
\end{equation}
In fact, for all $N\geq 100$,
\begin{equation}
  D(\mu_{\delta})\geq c_{3}\log(N),\quad\mbox{and}\quad
  T_{{\cal{A}}_{QS}(p(N))}(\mu_{\delta})\leq c_{4}(\log(N))^{1-q},
\end{equation}
for constants $c_{3},c_{4}>0,$ where we can take $c_{3}=\frac{3}{8}$
and $c_{4}=4440e^{2}<32810$. 
\end{cor}

\begin{proof}
  The proof is the same as that for Corollary 1 except that we now
  take $p=p(N)=\frac{1}{\left(\log(N)\right)^{q}}$. 
\end{proof}

\section*{References}

\noindent [1] Cohn I., Fonseca de Oliveira, A. and Buksman, E. (2016): Grover's search with local and total depolarizing channel errors. \itshape\,International Journal of Quantum Information,\upshape Volume 14, Issue 2, March 2016.

\noindent [2] Grover, L. (1996): A fast quantum mechanical algorithm for database search. In\itshape\,Proceedings of the twenty-eighth annual ACM symposium on theory of computing (STOC '96),\upshape 212-219. 

\noindent [3] Grover, L. (1997): Quantum mechanics helps in searching for a needle in a haystack.\itshape\,Phys. Rev. Lett.,\upshape\,79,2, 14 July 1997.

\noindent [4] Montanaro, A. (2010): Quantum search with advice. In\itshape\,Proc. of the 5th Conference on the Theory of Quantum Computation, Communication, and Cryptography (TQC'10).\upshape

\noindent [5] Motwani, R. and Raghavan, P. (1995): \itshape Randomized Algorithms. \upshape\,Cambridge: Cambridge
University Press.

\noindent [6] Nielsen, M. and Chuang, I. (2000): \itshape Quantum Computation and Computation Information. \upshape\,Cambridge: Cambridge
University Press.

\noindent [7] Preskill, J. (2013): Quantum Computing and the Entanglement Frontier.\itshape\, Bull. Am. Phys. Soc.,\upshape\,58.

\noindent [8] Shor, P. (1994): Algorithms for quantum computation: discrete logarithms and factoring. In \itshape\,35th Annual Symposium on Foundations of Computer Science\upshape, IEEE Computer Society Press, 124-134.

\noindent [9] Vrana, P., Reeb, D., Reitzner, D., and Wolf, M. (2014): Fault-ignorant quantum search.\itshape\,New Journal of Physics\upshape\, Vol. 16, July 2014.

\end{document}